\documentclass[a4paper,USenglish,cleveref,autoref,thm-restate,pdfa]{lipics-v2021}
\hideLIPIcs\def\subjclassHeading{}\ccsdesc{}\global\renewcommand\ccsdesc[2][100]{}
\nolinenumbers %
\title{A Note on the Complexity of Defensive Domination}
\allowdisplaybreaks[1]

\author{Steven Chaplick}
{Maastricht University}
{s.chaplick@maastrichtuniversity.nl}
{https://orcid.org/0000-0003-3501-4608}
{}

\author{Grzegorz Gutowski}
{Institute of Theoretical Computer Science, Faculty of Mathematics and Computer Science, Jagiellonian University, Krak{\'o}w, Poland}
{grzegorz.gutowski@uj.edu.pl}
{https://orcid.org/0000-0003-3313-1237}
{Partially supported by grant no.~2023/49/B/ST6/01738 from National Science Centre, Poland.}

\author{Tomasz Krawczyk}
{Faculty of Mathematics and Information Science, Warsaw University of Technology}
{tomasz.krawczyk@pw.edu.pl}
{https://orcid.org/0000-0002-8777-269X}
{Partially supported by grant no. 2024/53/B/ST6/02558  from National Science Centre, Poland.}

\authorrunning{S. Chaplick, G. Gutowski, T. Krawczyk}
\Copyright{Steven Chaplick, Grzegorz Gutowski, Tomasz Krawczyk}

\usepackage{algorithm}
\usepackage{algorithmicx}
\usepackage{algpseudocode}
\usepackage{hyperref}
\usepackage{graphicx}
\usepackage{tikz}
\usetikzlibrary{arrows}
\usetikzlibrary{patterns}
\usetikzlibrary{arrows}
\usetikzlibrary{patterns}
\usetikzlibrary{calc}
\usetikzlibrary{shapes.geometric}
\usetikzlibrary{shapes,snakes}
\usetikzlibrary{fit}
\usetikzlibrary{backgrounds}
\usetikzlibrary{decorations.pathmorphing}
\usetikzlibrary{decorations.pathreplacing}

\usepackage[textsize=small]{todonotes}
\usepackage{enumitem}
\usepackage{amsmath,amsfonts,amssymb}
\usepackage{xspace}

\newtheorem{probspecint}{Problem}
\newcommand{\newprobspec}[3]{\begin{samepage}\begin{probspecint}{#1}\begin{itemize}[leftmargin=6em]
\item[$\mathsf{Input:}$] #2
\item[$\mathsf{Output:}$] #3
\end{itemize}\end{probspecint}\end{samepage}}

\let\leq\leqslant
\let\geq\geqslant
\let\setminus\smallsetminus

\let\rho\varrho
\let\implies\Rightarrow

\newcommand{\POL}{\ensuremath{\mathsf{P}}\xspace}
\newcommand{\NP}{\ensuremath{\mathsf{NP}}\xspace}
\newcommand{\CONP}{\ensuremath{\mathsf{co}\text{-}\mathsf{NP}}\xspace}
\newcommand{\FPT}{\ensuremath{\mathsf{FPT}}\xspace}

\newcommand{\W}[1]{\ensuremath{\mathsf{W[#1]}}\xspace}
\newcommand{\PHS}[1]{\ensuremath{\Sigma^\mathsf{P}_{#1}}\xspace}

\newcommand{\SPTWO}{\PHS{2}}

\newcommand{\Pdom}{\textsc{DominatingSet}\xspace}
\newcommand{\Pdefdom}{\textsc{DefensiveDominatingSet}\xspace}
\newcommand{\Pkdefdom}{$k$\textsc{-DefensiveDominatingSet}\xspace}
\newcommand{\Padefdom}{$\mathcal{A}$\textsc{-DefensiveDominatingMultiset}\xspace}
\newcommand{\Pmultidom}{\textsc{DefensiveDominatingMultiset}\xspace}

\newcommand{\Pcliquedel}{\textsc{CliqueNodeDeletion}\xspace}
\newcommand{\PSPTWO}{\textsc{Existential-}2\textsc{-Level-}3{-CNF}\xspace}
\newcommand{\PSPTWOSAT}{\textsc{Existential-}2\textsc{-Level-SAT}\xspace}
\newcommand{\Pcheckdom}{\textsc{GoodDefense}\xspace}
\newcommand{\Pcocheckdom}{\textsc{BadDefense}\xspace}
\newcommand{\Phallset}{\textsc{HallSet}\xspace}
\newcommand{\Pclique}{\textsc{Clique}\xspace}

\newcommand{\brac}[1]{\left(#1\right)}
\newcommand{\sbrac}[1]{\left[#1\right]}
\newcommand{\set}[1]{\left\{#1\right\}}
\newcommand{\norm}[1]{\left|#1\right|}
\newcommand{\Oh}[1]{O\brac{#1}}

\DeclareMathOperator{\setcount}{count}
\DeclareMathOperator{\ispan}{sum}
\DeclareMathOperator{\tspan}{span}
\DeclareMathOperator{\leftb}{left}
\DeclareMathOperator{\rightb}{right}

\begin{document}

\maketitle

\begin{abstract}
In a graph~$G$, a \emph{$k$-attack}~$A$ is any set of at most $k$~vertices and \emph{$\ell$-defense}~$D$ is a set of at most $\ell$~vertices.
We say that defense~$D$ \emph{counters} attack~$A$ if each $a \in A$ can be matched to a distinct \emph{defender} $d \in D$ with~$a$ equal to~$d$ or $a$~adjacent to~$d$ in~$G$.
In the \emph{defensive domination problem}, we are interested in deciding, for a graph~$G$ and positive integers~$k$ and~$\ell$ given on input, if there exists an~$\ell$-defense that counters every possible $k$-attack on~$G$.
Defensive domination is a natural resource allocation problem and can be used to model network robustness and security, disaster response strategies, and redundancy designs.

The defensive domination problem is naturally in the complexity class~$\SPTWO$.
The problem was known to be \NP-hard in general, and polynomial-time algorithms were found for some restricted graph classes.
In this note, we prove that the defensive domination problem is~\SPTWO-complete.

We also introduce a natural variant of the defensive domination problem in which the defense is allowed to be a multiset of vertices.
This variant is also~\SPTWO-complete, but we show that it admits a polynomial-time algorithm in the class of interval graphs.
A similar result was known for the original setting in the class of proper interval graphs.
\end{abstract}

\section{Introduction}

All graphs discussed in this paper are finite and simple. 
The vertex set and edge set of a graph~$G$ are denoted by~$V(G)$ and~$E(G)$.
For a subset~$U\subseteq V(G)$, $G[U]$ denotes the subgraph of~$G$ \emph{induced by~$U$}, 
and $G \setminus U$ denotes the subgraph~$G[V(G)\setminus U]$, which is shortened to~$G \setminus v$ when $U = \set{v}$.
The \emph{neighborhood} of a vertex~$v$, denoted by~$N_{G}(v)$, comprises of vertices adjacent to~$v$ and the \emph{closed neighborhood} of~$v$ is $N_{G}[v] = N_{G}(v) \cup \set{v}$.
The \emph{closed neighborhood} and the \emph{neighborhood} of a set~$U\subseteq V(G)$ of vertices are defined as~$N_{G}[U] = \bigcup_{v \in U} N_{G}[v]$ and~$N_{G}(U) =  N_{G}[U] \setminus U$, respectively.
The subscript~$G$ can be dropped if the graph is clear from the context.
A graph on $t$ vertices with all $t \choose 2$ edges is a \emph{clique} $K_t$.
When $G[U]$ is isomorphic to $K_t$ for some set $U$ of vertices, we say that $G$ contains $K_t$ as a subgraph.
The clique problem asks for the maximum $t$ such that $K_t$ is contained as a subgraph in a given graph.
A set~$D\subseteq V(G)$ is a \emph{dominating set} when $N_G[D] = V(G)$.
The graph domination problem asks for the minimum size of a dominating set in a given graph and is a classical problem in graph theory and combinatorial optimization.
We use the following specification of the problem.
\newprobspec{\Pdom}{A graph $G$ and a positive integer $\ell$}{$\mathsf{Yes}$ if and only if $G$ admits a dominating set of size at most $\ell$}
This problem has broad practical applications in resource allocation, network design, analysis and security.
It is also of theoretical interest, as it is one of the first problems known to be \NP-complete, see~\cite{GareyJ79}, and is used as a base for countless reductions.
Consult~\cite{HaynesHH20,HaynesHH21,HaynesHH23} for various versions and applications of the domination problem.

Graph domination can be understood through the analogy in which the vertices of a graph are under threat of some attack and \emph{defenders} need to be placed in the vertices so that each vertex either has a defender stationed directly in it or in an adjacent vertex.
This concept is useful in network security, facility location problems (positioning service centers), and disaster response strategies (deploying rescue teams).
Presented in this way, the \Pdom problem looks for a minimum number of defenders that can counter any attack on a single vertex. 

Farley and Proskurowski~\cite{FarleyP04} proposed the following extension of the problem, called \emph{defensive domination}, where we prepare for a simultaneous attack on at most $k$ vertices.
We say that any set~$A$ of at most $k$ distinct vertices in $G$ is a \emph{$k$-attack}.
An \emph{$\ell$-defense}~$D$ is a set of at most $\ell$~distinct vertices of~$G$ and corresponds to placing $\ell$ defenders, one in each vertex of~$D$.
We say that defense~$D$ \emph{counters} attack~$A$ if there is a matching between $A$ and $D$ such that each $a \in A$ is matched to a distinct defender $d \in D$ with~$a$ equal to~$d$ or $a$~adjacent to~$d$ in~$G$.
A defense~$D$ that counters every possible $k$-attack in~$G$ is called a \emph{$k$-defensive dominating set}.
This extension is natural and meets the redundancy requirements usual for all applications of the domination problem. 
We use the following formal specification of the problem. 
\newprobspec{\Pdefdom}{A graph $G$ and positive integers $k$ and $\ell$}{$\mathsf{Yes}$ if and only if $G$ admits an $\ell$-defense that counters every $k$-attack in $G$}
The parametrized version of the problem, where the size of the attack is an external parameter and not a part of the input is also of interest.
\newprobspec{\Pkdefdom}{A graph $G$ and a positive integer $\ell$}{$\mathsf{Yes}$ if and only if $G$ admits an $\ell$-defense that counters every $k$-attack in $G$}

Observe that \Pkdefdom for $k=1$ is exactly \Pdom.
Dereniowski, Gavenčiak, and Kratochvíl~\cite{DereniowskiGK19} proposed a further extension of the problem, which seems a bit technical, but was successfully applied to study a variant of the cops and robbers game.
The proposed extension of defensive domination allows the placement of multiple defenders in a single vertex of the graph and limits the possible attacks to the ones that are explicitly specified on the input.
A \emph{multiset $\ell$-defense}~$D$ places $\ell$ defenders in total at the vertices of~$G$, each vertex getting as many defenders as its multiplicity in $D$.
Multiset defense~$D$ counters attack~$A$ if each $a \in A$ can be matched to a distinct defender stationed in $a$, or any vertex adjacent to $a$.
The formal specification of the problem proposed by Dereniowski, Gavenčiak, and Kratochvíl~\cite{DereniowskiGK19} follows.
\newprobspec{\Padefdom}{A graph $G$, a set of attacks $\mathcal{A} \subseteq 2^{V(G)}$, multisets $D_1$ and $D_2$ of vertices of $G$, and a positive integer $\ell$}{$\mathsf{Yes}$ if and only if $G$ admits a multiset $\ell$-defense $D$ with $D_1 \subseteq D \subseteq D_2$ that counters every attack $A \in \mathcal{A}$}
Observe, that we do not allow for multiset attacks, as it would lead to a different problem.

We believe that allowing for multiset defenses is an interesting extension and allows for various applications.
We propose the following natural extension of the defensive domination problem, that allows for multiset defenses, but drops other technical conditions introdued by Dereniowski, Gavenčiak, and Kratochvíl.
\newprobspec{\Pmultidom}{A graph $G$ and positive integers $k$ and $\ell$}{$\mathsf{Yes}$ if and only if $G$ admits a multiset $\ell$-defense that counters every $k$-attack in $G$}
To exemplify the strength of this extension, note that $2$-defensive dominating set on $K_{1,t}$ has at least $t$ defenders, but it is enough to use $2$ defenders in the multiset setting.
In the proof of the main result of this paper, we focus on \Pdefdom, but the hardness also applies to \Pmultidom, which we believe should attract more attention.

As \Pdom is \NP-complete, it is straightforward that all mentioned domination problems are \NP-hard.
It is also easy to observe that \Pkdefdom and \Padefdom are in fact \NP-complete.
On the other hand, as \Pdefdom is naturally expressed as:
$$
    \exists_{D\subseteq V(G),\norm{D}\leq\ell}:\quad\forall_{A\subseteq V(G),\norm{A}\leq k}:\quad \text{$D$ counters $A$ in $G$}\text{,}
$$
we easily get that \Pdefdom is in the second level of polynomial hierarchy class \SPTWO.
Consult textbook by Arora and Barak~\cite[Chapter~5]{AroraB09} for an introduction of the polynomial hierarchy.
Schaefer and Umans~\cite{SchaeferU02_1,SchaeferU02_2,SchaeferU08} give an extensive list of complete problems for different classes in the polynomial hierarchy.
For a very brief introduction, \SPTWO is defined as $\NP^\NP$ -- a class of languages decidable in polynomial time by nondeterministic Turing machines with access to \NP-oracle,
where \NP-oracle allows to test any language in \NP in a single step of execution.
The canonical complete problem for \SPTWO is the following.
\newprobspec{\PSPTWOSAT}{Formula $\varphi(x_1,\ldots,x_a,y_1,\ldots,y_b)$ with variables in two disjoint sets $\set{x_1,\ldots,x_a}$ and $\set{y_1,\ldots,y_b}$}{$\mathsf{Yes}$ if and only if the following Boolean formula is true.
    $$
        \exists_{x_1,x_2,\ldots,x_a}:\quad \forall_{y_1,y_2,\ldots,y_b}:\quad \varphi(x_1,\ldots,x_a,y_1,\ldots,y_b)
    $$}
It was indpendently proved by Stockmeyer~\cite{Stockmeyer76} and Wrathall~\cite{Wrathall76} that the class \SPTWO is exactly the class of languages reducible to \PSPTWOSAT via polynomial-time many-one reductions.
Clearly, the following nondeterministic algorithm using \NP-oracle solves \Pdefdom: algorithm first guesses a set of defenders $D$ and then uses \NP-oracle to test whether there exists an attack of size at most $k$ not defended by $D$.
This fact makes \SPTWO a natural complexity class for \Pdefdom problem.

Ekim, Farley, and Proskurowski~\cite{EkimFP20} showed that \Pdefdom is unlikely to be in \NP.
The reason is that the following problem that corresponds to checking if a given defense $D$ counters every $k$-attack on~$G$ is already \CONP-complete.
\newprobspec{\Pcheckdom}{A graph $G$, a subset $D$ of vertices, and a positive integer $k$}{$\mathsf{Yes}$ if and only if defense $D$ counters every $k$-attack in $G$}

For any defense set, or multiset, $D$, and any set of vertices $X \subseteq V(G)$, let $\setcount_D(X)$ denote the number of elements (counting multiplicities for multisets) in $D \cap X$.
The following connection between defensive domination and Hall's condition was already observed in~\cite{EkimFP20}.
\begin{observation}[Ekim, Farley, Proskurowski~\cite{EkimFP20}]\label{obs:hall} The following conditions are equivalent:
\begin{itemize}
    \item Defense $D$ counters every $k$-attack in~$G$.
    \item For every $k$-attack $A$ we have $\norm{A} \leq \setcount_D(N[A])$.
\end{itemize}
\end{observation}

This draws our attention to the complementary problem of \Pcheckdom and a very similar problem that is known to be \NP-complete and \W{1}-hard.

\newprobspec{\Pcocheckdom}{A graph $G$, a subset $D$ of vertices, and a positive integer $k$}{$\mathsf{Yes}$ if there exists $k$-attack $A$ with $\norm{A} > \setcount_D(N[A])$}
\newprobspec{\Phallset}{A bipartite graph $G$ with bipartition classes $U$ and $W$ and a positive integer $k$}{$\mathsf{Yes}$ if and only if there exists $X \subseteq U$ with $\norm{N_G(X)} < \norm{X}\leq k$}

You can find a parametrized reduction from \Pclique to \Phallset in~\cite[Exercise~13.28]{CyganFKLMPPS15}.
\Phallset has an easy parametrized reduction to \Pcocheckdom.
All of these observations allow for the following conclusion.
\begin{lemma}[Theorem~2.3 in Ekim, Farley, Proskurowski~\cite{EkimFP20}]\label{lem:w1}
\Pcheckdom is \CONP-complete. \Pcocheckdom is \NP-complete and \W{1}-hard when parametrized by $k$.
\end{lemma}

When $k$ is an external parameter of the problem, \Pkdefdom is in \NP, and it remains \NP-complete even when the input graph is restricted to split graphs~\cite{EkimFP20}, or bipartite graphs~\cite{HenningPT25}.
On the positive side, \Pdefdom admits polynomial-time algorithms when the input graph is restricted to cliques, cycles, trees~\cite{FarleyP04}, co-chain graphs, threshold graphs~\cite{EkimFP20}, or proper interval graphs~\cite{EkimFPS23}.

The main result of this paper is the following.
\begin{restatable}{theorem}{thmmain}\label{thm:main}
\Pdefdom and \Pmultidom are \SPTWO-complete.
\end{restatable}

The introduced multiset setting not only may better fit some applications, but might also be more approachable algorithmically.
For example, in \cref{sec:interval} we investigate the multiset defensive domination problem on the class of interval graphs.
A graph $G$ is an \emph{interval graph} when each vertex $v \in V(G)$ corresponds to a closed interval $I_v \subseteq \mathbb{R}$, and $\set{u,v} \in E(G)$ if and only if $I_u \cap I_v \neq \emptyset$.
We prove the following.
\begin{restatable}{theorem}{thminterval}\label{thm:interval}
\Pmultidom is in $\POL$ when the input graphs are restricted to the class of interval graphs.
\end{restatable}
A similar result for \Pdefdom was shown for proper interval graphs by Ekim, Farley, Proskurowski, and Shalom~\cite{EkimFPS23}.
The complexity of \Pdefdom for interval graphs remains unknown.

The proof of the main theorem is based on a reduction of the following problem, which was shown to be \SPTWO-complete by Rutenburg~\cite{Rutenburg86}.
\newprobspec{\Pcliquedel}{A graph $G$ and positive integers $s$ and $t$}{$\mathsf{Yes}$ if and only if $G$ admits a set $X$ of at most $s$ vertices such that $G\setminus X$ does not contain $K_t$ as a subgraph}
\begin{restatable}[Theorem~6 in Rutenburg~\cite{Rutenburg86}]{theorem}{thmrutenburg}\label{thm:rutenburg}
\Pcliquedel is \SPTWO-complete.
\end{restatable}
As the original paper includes only an idea of the proof that requires some minor alterations, we have decided to present a streamlined proof of \cref{thm:rutenburg} in Appendix~\ref{sec:clique}.

\section{Main Result}

We are ready for the proof of the main result.
\thmmain*
\begin{proof}
We present a reduction from \Pcliquedel.
Assume that we are given an instance $G,s,t$, where $s$ is the number of vertices to remove from $G$ and $t$ is the size of clique to avoid as a subgraph.
For technical reasons, we assume that $t \geq 4$.
Let $n =\norm{V(G)}$ denote the number of vertices in $G$.

We construct an equivalent instance $G',k,\ell$ of \Pdefdom.
We set the maximum size of an attack $k=n+s$, the maximum size of a defense $\ell=4(n+s)+ nt - (t-1)$, and construct the graph $G'$ as depicted in Figure~\ref{fig:reduction}:
\begin{itemize}
\item For each vertex $v \in V(G)$, we introduce two vertices $v'$ and $v''$ representing $v$ in $G'$.
Set $V$ denotes vertices $v'$, $v''$ introduced for all $v \in V(G)$.
\item For each edge $e=(u,v)$ in $E(G)$, we introduce the vertex $e'$ and add the edges joining $e'$ with four vertices $u'$, $u''$, $v'$, $v''$ in $V$.
Set $E$ denotes vertices $e'$ introduced for all $e \in E(G)$.
\item We introduce four independent sets: $I_1$ of size $n+s$; $I_2$ of size $n+s-{t \choose 2}$; $I_3$ of size $n+s+\ell$; and $I_4$ of size $n+s$.
\item We introduce three cliques: $Q_1$ of size $n+s$; $Q_2$ of size $n+s-(t+1)$; and $Q_4$ of size $n+s$.
\item For each vertex $v \in V(G)$, we introduce a complete bipartite graph with one bipartition class $I_v$ of size ${t \choose 2}$,
and the other class $I'_v$ of size $t$.
Set $I_V$ denotes the sum of sets $I_v$ for all $v \in V(G)$ and $I'_V$ denotes the sum of sets $I'_v$ for all $v \in V(G)$.
\item We add edges of complete bipartite graphs given by the biparition classes: $(I_1,Q_1)$; $(Q_2,Q_1\cup I_2\cup E\cup I_V)$; $(V,Q_4\cup I_3)$; $(Q_4,I_4)$; and $(\set{v',v''},I_v)$ for each $v\in V(G)$.
\end{itemize}

\tikzstyle{cross}  = [{path picture={ 
		\draw[black]
		(path picture bounding box.south east) -- (path picture bounding box.north west) (path picture bounding box.south west) -- (path picture bounding box.north east);
	}}]

\begin{figure}[htp!]
\begin{center}

\begin{tikzpicture}[scale=0.5]
			
\begin{scope}[shift={(0,3.75)}]
	\coordinate (K1_center) at (0,0) {};
	\coordinate (K1_label) at ($(K1_center)+(225:2)$) {};
	\coordinate (K11) at ($(K1_center)+(0:1)$) {};
	\coordinate (K12) at ($(K1_center)+(72:1)$) {};
	\coordinate (K13) at ($(K1_center)+(144:1)$) {};
	\coordinate (K14) at ($(K1_center)+(216:1)$) {};
	\coordinate (K15) at ($(K1_center)+(288:1)$) {};

	\coordinate (K1B1) at ($(K1_center)+(240:1.6)$) {};
	\coordinate (K1B2) at ($(K1_center)+(270:1.6)$) {};
	\coordinate (K1B3) at ($(K1_center)+(300:1.6)$) {};

	\coordinate (K1U1) at ($(K1_center)+(60:1.6)$) {};
	\coordinate (K1U2) at ($(K1_center)+(90:1.6)$) {};
	\coordinate (K1U3) at ($(K1_center)+(120:1.6)$) {};

	\draw [thick, decorate, decoration = {brace}] (-1.9,-1.2) --  (-1.9,1.2);

	\coordinate (K1_br_lab) at (-3,0) {};
\end{scope}

\begin{scope}[shift={(0,-1)}]
		\coordinate (K2_center) at (0,0) {};
		\coordinate (K2_label) at ($(K2_center)+(225:2)$) {};
		\coordinate (K21) at ($(K2_center)+(0:1)$) {};
		\coordinate (K22) at ($(K2_center)+(72:1)$) {};
		\coordinate (K23) at ($(K2_center)+(144:1)$) {};
		\coordinate (K24) at ($(K2_center)+(216:1)$) {};
		\coordinate (K25) at ($(K2_center)+(288:1)$) {};
		
		\coordinate (K2U1) at ($(K2_center)+(60:1.6)$) {};
		\coordinate (K2U2) at ($(K2_center)+(90:1.6)$) {};
		\coordinate (K2U3) at ($(K2_center)+(120:1.6)$) {};

		\coordinate (K2L1) at ($(K2_center)+(150:1.6)$) {};
		\coordinate (K2L2) at ($(K2_center)+(180:1.6)$) {};
		\coordinate (K2L3) at ($(K2_center)+(210:1.6)$) {};

		\coordinate (K2R1) at ($(K2_center)+(-30:1.6)$) {};
		\coordinate (K2R2) at ($(K2_center)+(0:1.6)$) {};
		\coordinate (K2R3) at ($(K2_center)+(30:1.6)$) {};

        \draw [thick, decorate, decoration = {brace}] (1.2,-1.9) --  (-1.2,-1.9);
		
		\coordinate (K2_br_lab) at (0,-2.5) {};
\end{scope}

\begin{scope}[shift={(12,-1)}]
	\coordinate (K4_center) at (0,0) {};
	\coordinate (K4_label) at ($(K4_center)+(225:2)$) {};
	\coordinate (K41) at ($(K4_center)+(0:1)$) {};
	\coordinate (K42) at ($(K4_center)+(72:1)$) {};
	\coordinate (K43) at ($(K4_center)+(144:1)$) {};
	\coordinate (K44) at ($(K4_center)+(216:1)$) {};
	\coordinate (K45) at ($(K4_center)+(288:1)$) {};
	
	\coordinate (K4L1) at ($(K4_center)+(120:1.6)$) {};
	\coordinate (K4L2) at ($(K4_center)+(150:1.6)$) {};
	\coordinate (K4L3) at ($(K4_center)+(180:1.6)$) {};
	\coordinate (K4L4) at ($(K4_center)+(210:1.6)$) {};
	
	\coordinate (K4R1) at ($(K4_center)+(30:1.6)$) {};
	\coordinate (K4R2) at ($(K4_center)+(0:1.6)$) {};
	\coordinate (K4R3) at ($(K4_center)+(330:1.6)$) {};
		
	\draw [thick, decorate, decoration = {brace}] (1.2,-1.9) --  (-1.2,-1.9);
	
	\coordinate (K4_br_lab) at (0,-2.5) {};

\end{scope}

\begin{scope}[shift={(-0.5,7.5)}]
	\draw[] (-2.5,-0.5) rectangle (2.5,0.5);
	\coordinate (I1_center) at (0,0) {};
	\coordinate (I1_label) at (-3,0) {};
	
	\coordinate (I11) at (-2,0) {};
	\coordinate (I12) at (-1,0) {};
	\coordinate (I13) at (0,0) {};
	\coordinate (I14) at (1,0) {};
	\coordinate (I15) at (2,0) {};
    \draw[dotted, thick] (-0.35,0)--(0.35,0) ;
	
	\coordinate (I1B1) at (-1.5,-0.6) {};
	\coordinate (I1B2) at (-0.5,-0.6) {};
	\coordinate (I1B3) at (0.5,-0.6) {};
	\coordinate (I1B4) at (1.5,-0.6) {};
	
	\coordinate (I1BL) at (-2.5,0.8) {};
	\coordinate (I1BR) at (2.5,0.8) {};
	
	\coordinate (I1_br_lab) at (0,1.3) {};
\end{scope}

\begin{scope}[shift={(-3,-0.5)}]
	\draw[] (-0.5,-2.5) rectangle (0.5,2.5);
	\coordinate (I2_center) at (0,0) {};
	\coordinate (I2_label) at (0,3) {};

	\coordinate (I21) at (0,-2) {};
	\coordinate (I22) at (0,-1) {};
	\coordinate (I23) at (0,0) {};
    \draw[dotted, thick] (0,-0.35)--(0,0.35) ;

	\coordinate (I24) at (0,1) {};
	\coordinate (I25) at (0,2) {};

	\coordinate (I2R1) at (0.6,1.5) {};
	\coordinate (I2R2) at (0.6,0.5) {};
	\coordinate (I2R3) at (0.6,-0.5) {};
	\coordinate (I2R4) at (0.6,-1.5) {};
    
	\coordinate (I2LB) at (-0.8,-2.5) {};
	\coordinate (I2RB) at (-0.8,2.5) {};
	\coordinate (I2_br_lab) at (-2.1,0) {};

	\draw [thick, decorate, decoration = {brace}] (I2LB) --  (I2RB);
	
\end{scope}

\begin{scope}[shift={(12,5.5)}]
	\draw[] (-0.5,-2.5) rectangle (0.5,2.5);
	\coordinate (I3_center) at (0,0) {};
	\coordinate (I3_label) at (0,3) {};
	
	\coordinate (I31) at (0,-2) {};
	\coordinate (I32) at (0,-1) {};
	\coordinate (I33) at (0,0) {};
	\coordinate (I34) at (0,1) {};
	\coordinate (I35) at (0,2) {};
    \draw[dotted, thick] (0,-0.35)--(0,0.35) ;
	
	\coordinate (I3L1) at (-0.6,1.5) {};
	\coordinate (I3L2) at (-0.6,0) {};
	\coordinate (I3L3) at (-0.6,-1.5) {};

    \coordinate (I3LB) at (0.8,2.5) {};
    \coordinate (I3RB) at (0.8,-2.5) {};
	\coordinate (I3_br_lab) at (2.5,0) {};
		
\end{scope}

\begin{scope}[shift={(15,-1)}]
	\draw[] (-0.5,-2.5) rectangle (0.5,2.5);
	\coordinate (I4_center) at (0,0) {};
	\coordinate (I4_label) at (0,3) {};
	
	\coordinate (I41) at (0,-2) {};
	\coordinate (I42) at (0,-1) {};
	\coordinate (I43) at (0,0) {};
    \draw[dotted, thick] (0,-0.35)--(0,0.35) ;
	\coordinate (I44) at (0,1) {};
	\coordinate (I45) at (0,2) {};
	
	\coordinate (I4L1) at (-0.6,1.5) {};
	\coordinate (I4L2) at (-0.6,0.5) {};
	\coordinate (I4L3) at (-0.6,-0.5) {};
	\coordinate (I4L4) at (-0.6,-1.5) {};

    \coordinate (I4BL) at (0.8,2.5) {};
    \coordinate (I4BR) at (0.8,-2.5) {};
	\coordinate (I4_br_lab) at (2.1,0) {};
	
\end{scope}

\begin{scope}[shift={(4.5,-0.5)}]
	\draw[] (-0.5,-3.5) rectangle (0.5,3.5);
	\coordinate (E_center) at (0,0) {};
	\coordinate (E_label) at (0,-4) {};
	
	\coordinate (E1) at (0,-3) {};
	\coordinate (E2) at (0,-2) {};
	\coordinate (E3) at (0,-1) {};
    \draw[dotted, thick] (0,-1.35)--(0,-0.65) ;
	\coordinate (E4) at (0,0) {};
	\coordinate (EE_label) at (-0.8,0) {};
    \draw[dotted, thick] (0,1.35)--(0,0.65) ;
	\coordinate (E5) at (0,1) {};
	\coordinate (E6) at (0,2) {};
	\coordinate (E7) at (0,3) {};
	
	\coordinate (EL1) at (-0.6,-2.5) {};
	\coordinate (EL2) at (-0.6,-0.75) {};
	\coordinate (EL3) at (-0.6,0.75) {};
	\coordinate (EL4) at (-0.6,2.5) {};

\end{scope}

\begin{scope}[shift={(8,-0.5)}]
	\draw[] (-0.65,-3.5) rectangle (0.65,3.5);
	\coordinate (V_center) at (0,0) {};
	\coordinate (V_label) at (0,-4) {};
    
	\draw[fill=white] (-0.5,-0.75) rectangle (0.5,-2.75);
	\coordinate (U1) at (0,-2.25) {};
	\coordinate (U2) at (0,-1.25) {};
	\coordinate (U1_label) at (1.05,-2.25) {};
	\coordinate (U2_label) at (1.05,-1.25) {};
	\coordinate (UL1) at (-0.4,-2.25) {};
	\coordinate (UL2) at (-0.4,-1.25) {};
	\draw[dotted, thick] (0,-0.4)--(0,0.4) ;

	\draw[] (-0.5,0.75) rectangle (0.5,2.75);
	\coordinate (V1) at (0,2.25) {};
	\coordinate (V2) at (0,1.25) {};
	\coordinate (V1_label) at (1.05,2.25) {};
	\coordinate (V2_label) at (1.05,1.25) {};
	\coordinate (VL1) at (-0.4,2.25) {};
	\coordinate (VL2) at (-0.4,1.25) {};
	
	\coordinate (VR1) at (0.72,-3) {};
	\coordinate (VR2) at (0.72,0) {};
	\coordinate (VR3) at (0.72,3) {};

\end{scope}

\begin{scope}[shift={(4.5,6.5)}]
	\draw[] (-0.5,-1.5) rectangle (0.5,1.5);
	\coordinate (Iv_center) at (0,0) {};
	\coordinate (Iv_label) at (-0.8,2) {};
	\coordinate (IV_label) at (0,3.2) {};
	\draw[dotted, thick] (0,-2.0)--(0,-2.7) ;
    \draw[dotted, thick] (0,2)--(0,2.7) ;

	\coordinate (Iv1) at (0,-1) {};
	\coordinate (Iv2) at (0,0) {};
    \draw[dotted, thick] (0,-0.35)--(0,0.35) ;

	\coordinate (Iv3) at (0,1) {};
	
	\coordinate (IvR1) at (0.6,1) {};
	\coordinate (IvR2) at (0.6,0) {};
	\coordinate (IvR3) at (0.6,-1) {};

	\coordinate (IvL1) at (-0.6,1) {};
	\coordinate (IvL2) at (-0.6,0) {};
	\coordinate (IvL3) at (-0.6,-1) {};

	\coordinate (IvLB) at (-0.8,-1.5) {};
	\coordinate (IvRB) at (-0.8,1.5) {};
	\coordinate (Iv_br_lab) at (-1.5,0) {};
		
\end{scope}

\begin{scope}[shift={(8,6.5)}]
	\draw[] (-0.5,-1.5) rectangle (0.5,1.5);
	\coordinate (Ivv_center) at (0,0) {};
	\coordinate (Ivv_label) at (0.8,2) {};
	\coordinate (IVV_label) at (0,3.2) {};
	\draw[dotted, thick] (0,-2.0)--(0,-2.7) ;
    \draw[dotted, thick] (0,2)--(0,2.7) ;
	
	\coordinate (Ivv1) at (0,-1) {};
	\coordinate (Ivv2) at (0,0) {};
	\coordinate (Ivv3) at (0,1) {};
	
    \draw[dotted, thick] (0,-0.35)--(0,0.35) ;

	\coordinate (IvvL1) at (-0.6,1) {};
	\coordinate (IvvL2) at (-0.6,0) {};
	\coordinate (IvvL3) at (-0.6,-1) {};
	
	\coordinate (IvvBL) at (0.8,1.5) {};
	\coordinate (IvvBR) at (0.8,-1.5) {};

    \coordinate (Ivv_br_lab) at (1.3,0) {};
	
\end{scope}

    \draw[-,lightgray] (K1B1)--(K2U1);
    \draw[-,lightgray] (K1B1)--(K2U2);
    \draw[-,lightgray] (K1B1)--(K2U3);
    \draw[-,lightgray] (K1B2)--(K2U1);
    \draw[-,lightgray] (K1B2)--(K2U2);
    \draw[-,lightgray] (K1B2)--(K2U3);
    \draw[-,lightgray] (K1B3)--(K2U1);
    \draw[-,lightgray] (K1B3)--(K2U2);
    \draw[-,lightgray] (K1B3)--(K2U3);

    \draw[-,lightgray] (I2R1)--(K2L1);
	\draw[-,lightgray] (I2R1)--(K2L2);
    \draw[-,lightgray] (I2R2)--(K2L1);
	\draw[-,lightgray] (I2R2)--(K2L2);
	\draw[-,lightgray] (I2R2)--(K2L3);
    \draw[-,lightgray] (I2R3)--(K2L1);
	\draw[-,lightgray] (I2R3)--(K2L2);
	\draw[-,lightgray] (I2R3)--(K2L3);
	\draw[-,lightgray] (I2R4)--(K2L2);
	\draw[-,lightgray] (I2R4)--(K2L3);

\draw[-,lightgray] (I1B1)--(K1U2);
\draw[-,lightgray] (I1B1)--(K1U3);
\draw[-,lightgray] (I1B2)--(K1U1);
\draw[-,lightgray] (I1B2)--(K1U2);
\draw[-,lightgray] (I1B2)--(K1U3);
\draw[-,lightgray] (I1B3)--(K1U1);
\draw[-,lightgray] (I1B3)--(K1U2);
\draw[-,lightgray] (I1B3)--(K1U3);
\draw[-,lightgray] (I1B4)--(K1U1);
\draw[-,lightgray] (I1B4)--(K1U2);
\draw[-,lightgray] (I1B4)--(K1U3);

\draw[-,lightgray] (K2R1)--(EL1);
\draw[-,lightgray] (K2R1)--(EL2);
\draw[-,lightgray] (K2R1)--(EL3);
\draw[-,lightgray] (K2R1)--(EL4);
\draw[-,lightgray] (K2R2)--(EL1);
\draw[-,lightgray] (K2R2)--(EL2);
\draw[-,lightgray] (K2R2)--(EL3);
\draw[-,lightgray] (K2R2)--(EL4);
\draw[-,lightgray] (K2R3)--(EL1);
\draw[-,lightgray] (K2R3)--(EL2);
\draw[-,lightgray] (K2R3)--(EL3);
\draw[-,lightgray] (K2R3)--(EL4);

\draw[-,lightgray] (K2R1)--(IvL2);
\draw[-,lightgray] (K2R1)--(IvL3);
\draw[-,lightgray] (K2R2)--(IvL1);
\draw[-,lightgray] (K2R2)--(IvL3);
\draw[-,lightgray] (K2R3)--(IvL1);
\draw[-,lightgray] (K2R3)--(IvL2);
\draw[-,lightgray] (K2R3)--(IvL3);

\draw[-,lightgray] (IvR1)--(IvvL1);
\draw[-,lightgray] (IvR1)--(IvvL2);
\draw[-,lightgray] (IvR1)--(IvvL3);
\draw[-,lightgray] (IvR2)--(IvvL1);
\draw[-,lightgray] (IvR2)--(IvvL2);
\draw[-,lightgray] (IvR2)--(IvvL3);
\draw[-,lightgray] (IvR3)--(IvvL1);
\draw[-,lightgray] (IvR3)--(IvvL2);
\draw[-,lightgray] (IvR3)--(IvvL3);

\draw[-,lightgray] (I3L1)--(VR1);
\draw[-,lightgray] (I3L1)--(VR2);
\draw[-,lightgray] (I3L1)--(VR3);
\draw[-,lightgray] (I3L2)--(VR1);
\draw[-,lightgray] (I3L2)--(VR2);
\draw[-,lightgray] (I3L2)--(VR3);
\draw[-,lightgray] (I3L3)--(VR1);
\draw[-,lightgray] (I3L3)--(VR2);
\draw[-,lightgray] (I3L3)--(VR3);

\draw[-,lightgray] (IvR1)--(V1);
\draw[-,lightgray] (IvR1)--(V2);
\draw[-,lightgray] (IvR2)--(V1);
\draw[-,lightgray] (IvR2)--(V2);
\draw[-,lightgray] (IvR3)--(V1);
\draw[-,lightgray] (IvR3)--(V2);

\draw[-,lightgray] (VR1)--(K4L2);
\draw[-,lightgray] (VR1)--(K4L3);
\draw[-,lightgray] (VR1)--(K4L4);
\draw[-,lightgray] (VR2)--(K4L2);
\draw[-,lightgray] (VR2)--(K4L3);
\draw[-,lightgray] (VR2)--(K4L4);

\draw[-,lightgray] (VR3)--(K4L1);
\draw[-,lightgray] (VR3)--(K4L2);
\draw[-,lightgray] (VR3)--(K4L3);

\draw[-,lightgray] (K4R1)--(I4L1);
\draw[-,lightgray] (K4R1)--(I4L2);
\draw[-,lightgray] (K4R1)--(I4L3);
\draw[-,lightgray] (K4R1)--(I4L4);
\draw[-,lightgray] (K4R2)--(I4L1);
\draw[-,lightgray] (K4R2)--(I4L2);
\draw[-,lightgray] (K4R2)--(I4L3);
\draw[-,lightgray] (K4R2)--(I4L4);
\draw[-,lightgray] (K4R3)--(I4L1);
\draw[-,lightgray] (K4R3)--(I4L2);
\draw[-,lightgray] (K4R3)--(I4L3);
\draw[-,lightgray] (K4R3)--(I4L4);

\draw[fill=lightgray!30] (K1_center) ellipse (1.5 and 1.5);
\draw[fill=lightgray!30] (K2_center) ellipse (1.5 and 1.5);
\draw[fill=lightgray!30] (K4_center) ellipse (1.5 and 1.5);

\begin{scope}[shift={(8,-0.5)}]
	\draw[fill=white] (-0.5,-0.75) rectangle (0.5,-2.75);
	\draw[fill=white] (-0.5,0.75) rectangle (0.5,2.75);
	\coordinate (V1) at (0,2.25) {};
	\coordinate (V2) at (0,1.25) {};
	\coordinate (VL1) at (-0.4,2.25) {};
	\coordinate (VL2) at (-0.4,1.25) {};
	
	\coordinate (VR1) at (0.7,-3) {};
	\coordinate (VR2) at (0.7,0) {};
	\coordinate (VR3) at (0.7,3) {};

\end{scope}

\draw[very thick,-] (E4)--(V1);
\draw[very thick,-] (E4)--(V2);
\draw[very thick,-] (E4)--(U1);
\draw[very thick,-] (E4)--(U2);

\tikzstyle{every node}=[thick, circle, minimum size=10pt,inner sep=0pt,draw, fill=white]

\node [cross]  at (K1_center) {};
\node [cross]  at (K11) {};
\node [cross]  at (K12) {};
\node [cross]  at (K13) {};
\node [cross]  at (K14) {};
\node [cross]  at (K15) {};

\node [cross]  at (K2_center) {};
\node [cross]  at (K21) {};
\node [cross]  at (K22) {};
\node [cross]  at (K23) {};
\node [cross]  at (K24) {};
\node [cross]  at (K25) {};
		
\node [cross]  at (K4_center) {};
\node [cross]  at (K41) {};
\node [cross]  at (K42) {};
\node [cross]  at (K43) {};
\node [cross]  at (K44) {};
\node [cross]  at (K45) {};

\node[cross]  at (V1) {};
\node[cross]  at (V2) {};
\node[]  at (U1) {};
\node[cross]  at (U2) {};

\node  at (Iv1) {};
\node  at (Iv3) {};

\node[cross]  at (Ivv1) {};
\node[cross]  at (Ivv3) {};

\node at (I11) {};
\node  at (I12) {};
\node  at (I14) {};
\node  at (I15) {};

\node at (I21) {};
\node  at (I22) {};
\node  at (I24) {};
\node  at (I25) {};

\node at (I31) {};
\node at (I32) {};
\node at (I34) {};
\node at (I35) {};

\node at (I41) {};
\node at (I42) {};
\node at (I44) {};
\node at (I45) {};

\node at (E1) {};
\node  at (E2) {};
\node  at (E4) {};
\node  at (E6) {};
\node  at (E7) {};

\draw [thick, decorate, decoration = {brace}] (I3LB) --  (I3RB);
\draw [thick, decorate, decoration = {brace}] (IvLB)--  (IvRB);
\draw [thick, decorate, decoration = {brace}] (IvvBL) --  (IvvBR);
\draw [thick, decorate, decoration = {brace}] (I4BL) --  (I4BR);
\draw [thick, decorate, decoration = {brace}] (I1BL) --  (I1BR);

\tikzstyle{every node}=[inner sep=1pt]
\node at (K1_label) {$Q_1$};
\node at (K2_label) {$Q_2$};
\node at (K4_label) {$Q_4$};
\node at (I1_label) {$I_1$};
\node at (I2_label) {$I_2$};
\node at (I3_label) {$I_3$};
\node at (I4_label) {$I_4$};
\node at (V_label) {$V$};
\node at (E_label) {$E$};
\node at (Iv_label) {$I_v$};
\node at (Ivv_label) {$I'_v$};
\node at (IV_label) {$I_V$};
\node at (IVV_label) {$I'_V$};

\node at (K1_br_lab) {$n + s$};
\node at (K2_br_lab) {$n + s - (t+1)$};
\node at (K4_br_lab) {$n + s$};
\node at (I2_br_lab) {$\begin{array}{c}n + s \\ - {t \choose 2}\end{array}$};
\node at (I1_br_lab) {$n + s$};
\node at (I4_br_lab) {$n + s$};
\node at (Iv_br_lab) {${t \choose 2}$};
\node at (Ivv_br_lab) {$t$};
\node at (I3_br_lab) {$n+s+\ell$};

\node at (V1_label) {$v'$};
\node at (V2_label) {$v''$};

\node at (U1_label) {$u''$};
\node at (U2_label) {$u'$};

\node at (EE_label) {$e'$};
	\end{tikzpicture}
\end{center}
	\caption{\label{fig:reduction} Reduction: The edge $e=(u,v)$ of $G$ is represented by the edges between $e'$ and $v',v'',u',u''$ in $G'$. The defenders are represented by crossed circles. 
	}
\end{figure}
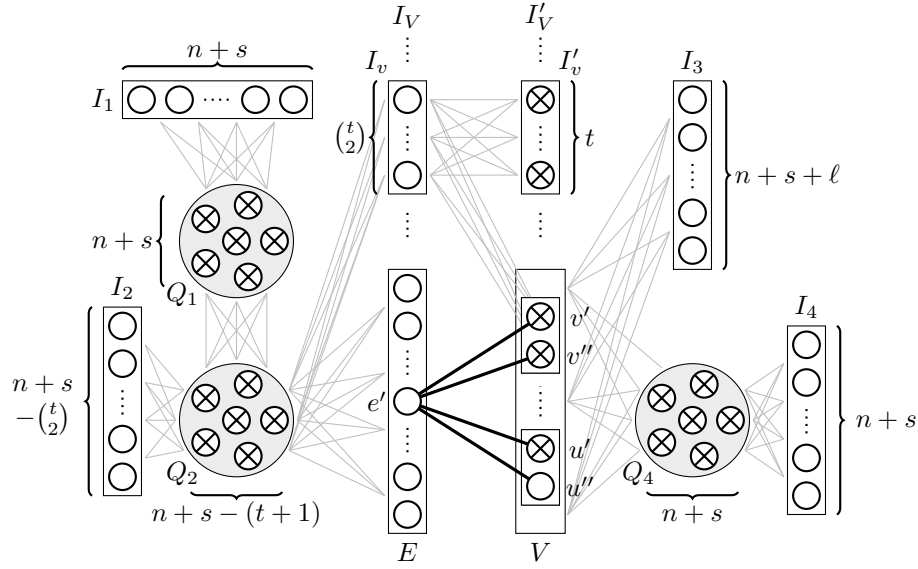

We claim that $s$ vertices can be removed from $G$ so that the resulting graph does not include $K_t$ as a subgraph if and only if there is an $\ell$-defense that counters every $k$-attack in $G'$.
We independently prove the implications in both directions.

($\Rightarrow$)
Let $X$ be a solution to the instance $G,s,t$.
We have $\norm{X} = s$ and $G\setminus X$ does not include $K_t$ as a subgraph.
We construct the following defense $D$:
$$D = Q_1 \cup Q_2 \cup Q_4 \cup I'_V  \cup \set{v': v \in V(G)} \cup \set{v'': v \in X}\text{.}$$
Note that $D$ positions exactly $n+s$ defenders on the vertices representing the vertices of $G$.
Additionally, $D$ has $nt$ defenders on $I'_V$ and $3(n+s)-(t+1)$ defenders on $Q_1 \cup Q_2 \cup Q_4$, which gives $\ell$ defenders in total.
We show that $D$ counters every attack of size at most $n+s$.
Suppose for a contradiction that $A$ is an inclusion minimal attack with $\setcount_D(N[A]) < \norm{A} \leq n+s$.
First, we observe that every vertex in the sets $I_1$, $Q_1$, $Q_2$, $Q_4$, $V$, $I_1$, $I_3$, and $I_4$ is adjacent to at least $n + s$ defenders (either from the set $Q_1$, $Q_2$, $Q_4$, or $V$).
If any of these vertices is included in $A$, then $\setcount_D(N[A]) \geq n+s \geq \norm{A}$.
We conclude that
$$
A \subseteq I_2 \cup E \cup I_V \cup I'_V\text{.}
$$
As every vertex in $I'_V$ is in~$D$, we do not have $A \subseteq I'_V$ and hence $A \cap (I_2 \cup E \cup I_V) \neq \emptyset$.
In particular, $Q_2 \subseteq N[A]$, and hence $\setcount_D(N[A]) \geq \norm{Q_2} = n+s - (t+1)$.
Since every vertex in $I_V$ is adjacent to additional $t+1$ defenders in $V \cup I'_V$, we have $A \cap I_V = \emptyset$.
Now, $N[I'_V] \cap A = I'_V \cap A$, and $I'_V \subseteq D$, so by the minimality of $A$ we get $A \cap I'_V = \emptyset$.
We conclude that $$A \subseteq I_2 \cup E\text{.}$$
We also have $A \cap E \neq \emptyset$, as $\norm{A} > n + s - (t+1) \geq n + s - {t \choose 2}=\norm{I_2}$, as $t+1\leq {t \choose 2}$ for $t\geq 4$.
As $A$ includes at least one vertex in $E$, $N[A]$ includes at least two vertices in $V\cap D$ and $\norm{A} > n + s - (t+1) + 2$.
We conclude that

$$\norm{A\cap E} \geq \norm{A} - \norm{I_2} > n + s - (t+1) + 2 - \brac{n + s - {t \choose 2}} = {t \choose 2} - (t-1) = {t-1 \choose 2}\text{.}$$

Now, as $A$ includes more than ${t-1 \choose 2}$ vertices in $E$, $N[A]$ includes at least $t$ vertices in $V\cap D$, and $\norm{A} > n + s - (t+1) + t$. 
Hence, $\norm{A}=n+s$ and $I_2 \subsetneq A$, as otherwise we would have $\setcount_D(N[A]) > n+s$.
We call an attack~$A$ to be \emph{serious} if

$$I_2 \subsetneq A \subseteq I_2 \cup E, \quad \norm{A \cap E} = {t \choose 2}.$$

We have shown that attack~$A$ is serious.
Any ${t \choose 2}$ edges in $G$ span at least $t+1$ vertices or span $t$ vertices that form a clique $K_t$ in $G$.
As we know that $G\setminus X$ does not contain $K_t$ as a subgraph, we get that if the edges span only $t$ vertices then at least one of them is in~$X$.
In either case, we get $\setcount_D(N[A]) \geq n + s - (t+1) + (t+1) = n+s$ and $D$ counters $A$.

\medskip

($\Leftarrow$)
Let $D$ be an $\ell$-defense that counters every $k$-attack in $G'$.
We make the following observations.
\begin{enumerate}
    \item As $D$ counters attack $I_1$, there are at least $n+s$ defenders in $I_1\cup Q_1$.
    \item As $D$ counters attack $I_4$, there are at least $n+s$ defenders in $I_4\cup Q_4$.
    \item\label{item:manydef} As $D$ counters every possible $(n+s)$-attack in $I_3\setminus D$, there are at least $n+s$ defenders in $V$.
    \item For every $v\in V(G)$, as $D$ counters attack $I'_v$, there are at least $t$ defenders in $I_v\cup I'_v$.
    Thus, in total, there are at least $nt$ defenders in $I_V\cup I'_V$.
    \item By calculation, there are at most $n+s-(t+1)$ defenders in $I_2\cup Q_2 \cup E$.
    \item\label{item:onedef} For every $v\in V(G)$, as $D$ counters attack $I_2\cup I_v$ and there are at most $n+s-(t+1)$ defenders in $I_2\cup Q_2$, there are at least $t+1$ defenders in $I_v\cup I'_v\cup \set{v',v''}$.
    As $\norm{I'_v}=t$, at least one of the defenders is in $I_v\cup\set{v',v''}$.
\end{enumerate}

We now construct a modified defense $D'$ and claim that $D'$ counters every serious attack on $G$.
As we focus on serious attacks, and the defenders in $I_V$ are not used to counter any serious attack, we move some of them.
For each $v \in V(G)$ if $v' \notin D$ and $v'' \notin D$, we move one defender from $I_v$ to $v'$ (guaranteed to be there by point~\ref{item:onedef}). 
Observe that even after this move there are at least $t$ defenders in $I_v\cup I'_v$.
As there are still at least $nt + (n+s) + (n+s)$ defenders in $I_V\cup I'_V\cup I_4\cup Q_4\cup I_1 \cup Q_1$, there are at most $(n+s) + (n+s-(t+1)$ in $I_2 \cup Q_2 \cup E \cup V$.
Second, since serious attacks include only vertices in $I_2\cup E$, and vertices in $Q_2$ dominate $I_2\cup E$, and there are at most $n+s-(t+1)$ defenders in $I_2\cup Q_2 \cup E$, we move all defenders from $I_2$ and $E$ to $Q_2$.
Third, while $\norm{D\cap V} > n+s$, we select any $v \in V(G)$ with defenders both in $v'$ and $v''$ and move the defender from $v''$ to $Q_2$. 
The resulting defense $D'$ has the property that it also counters every serious attack, as $D$ did.

The resulting defense $D'$ has exactly $n+s$ defenders in $V$, at most $n+s-(t+1)$ defenders in $Q_2$, at least one defender in each $\set{v',v''}$ for every $v \in V(G)$, and counters every serious attack.
Let $X$ be a set of vertices $v$ in $V(G)$ for which both $v'$ and $v''$ are in $D'$. There are exactly $s$ such vertices.
We claim that $G\setminus X$ does not contain $K_t$ as a subgraph.
Indeed, given a $t$-element set $Q$ of vertices of $G$ with $Q\cap X = \emptyset$ and $G[Q]$ isomorphic to $K_t$, we can create a $k$-attack composed of $I_2$ and the set of ${t \choose 2}$ vertices $e'$ representing edges of $G[Q]$.
This attack has size $n+s$ and has at most $n+s-(t+1)$ neighboring defenders in $Q_2$ and exactly $t$ neighboring defenders in $V$.
Thus, this serious attack is not countered by $D'$, which contradicts the construction of~$D'$.

We leave it to the reader to verify that exactly the same reduction also shows that \Pmultidom is also a \SPTWO-complete problem.
\end{proof}

\section{Interval Graphs}\label{sec:interval}
Building on the work of Ekim, Farley, Proskurowski, and Shalom~\cite{EkimFPS23}, who presented a greedy algorithm for \Pdefdom on proper interval graphs, we develop a similar greedy strategy for \Pmultidom on general interval graphs.

For the remainder of this section, let $G$ be an interval graph given by its interval representation: each vertex $v \in V(G)$ corresponds to a closed interval $I_v \subseteq \mathbb{R}$, and $\set{u,v} \in E(G)$ if and only if $I_u \cap I_v \neq \emptyset$.
We assume that this representation ensures that no two distinct intervals share an endpoint.
For any bounded closed set $S \subseteq \mathbb{R}$, let $\leftb(S)$ and $\rightb(S)$ denote the minimum and the maximum element in $S$, respectively. 
For any set (or multiset) of intervals $X$, let $\ispan(X) = \bigcup_{I \in X} I$ denote their union, and let $\tspan(X)$ be the minimum closed interval containing every interval in $X$, that is, the interval $\sbrac{\leftb(X),\rightb(X)}$.
For a set $Y \subseteq V(G)$, we define $\ispan(Y) = \ispan(\set{I_v : v \in Y})$ and $\tspan(Y) = \tspan(\set{I_v : v \in Y})$. 
A set (or multiset) of intervals $X$ is \emph{proper} if no interval in $X$ is a proper subset of another interval in $X$.
We say that a set $Y \subseteq V(G)$ is proper, when $\set{I_v : v \in Y}$ is proper.
The algorithm presented by Ekim, Farley, Proskurowski, and Shalom~\cite{EkimFPS23} for \Pdefdom worked under the condition that $V(G)$ is proper.

Our first observation is that in the multiset setting we can focus on constructing proper defenses.
\begin{observation}\label{obs:proper}
    For any multiset defense $D$, there exists a proper multiset defense $D'$ such that $\norm{D'} = \norm{D}$ and $D'$ counters any attack that $D$ counters. 
\end{observation}
\begin{proof}
    Consider a multiset defense $D$. If it is not proper, then there exist two vertices $u, v \in D$ such that $I_u \subsetneq I_v$.
    Let $D'$ be the multiset defense obtained from $D$ by replacing every copy of $u$ with an additional copy of $v$.
    Clearly, $\norm{D'} = \norm{D}$.
    Since $I_u \subsetneq I_v$, we have that $N[u] \subseteq N[v]$ and that any attack $A$ that is countered by $D$ using defenders in $u$ is countered by $D'$ using added defenders in $v$.
    Therefore, $D'$ counters every attack countered by $D$.
    Observe that this replacement increases the total length of the intervals that represent vertices in the defense.
    Thus, repeating this replacement procedure eventually stops and yields a proper defense that counters any attack that the original defense counters.
\end{proof}

Note that \cref{obs:proper} holds specifically for the multiset setting and does not have a direct analogue for \Pdefdom.
The next observation is that sets with smaller union are more dangerous for any defense than those with larger union.
\begin{observation}\label{obs:span}
    For any defense $D$ and two sets $A_1$, $A_2$ with $\setcount_D(N[A_1]) < \norm{A_1}$, $\norm{A_1} \le \norm{A_2}$, and $\ispan(A_2) \subseteq \ispan(A_1)$ we have $\setcount_D(N[A_2]) < \norm{A_2}$.
\end{observation}
\begin{proof}
    We have that $\setcount_D(N[A_2])$ is the number of intervals in $D$ that have a nonempty intersection with $\ispan(A_2)$ which is a subset of $\ispan(A_1)$.
    Thus, $$\setcount_D(N[A_2]) \le \setcount_D(N[A_1]) < \norm{A_1} \le \norm{A_2}\text{.}$$
\end{proof}

Let $x=\rightb(I_v)$ for some vertex $v \in V(G)$.
Let $V_x = \set{ v \in V(G) : \rightb(I_v) \le x}$ denote the nonempty set of vertices that lie completely to the left of $x$ in the representation.
Let $c = \norm{V_x}$ and for every integer $1 \le i \le c$ we define the \emph{$i$-th block} at $x$, denoted $B_{x,i}$, in the following way.
Let $v_1,v_2,\ldots,v_c$ be $V_x$ arranged in a sequence sorted descending by the left endpoint of the representing interval, that is, $\leftb(I_u) > \leftb(I_w)$ if and only if $u$ appears earlier than $w$ in the sequence.
We select $B_{x,i} = \set{v_1,v_2,\ldots,v_i}$ to be the first $i$ elements in the sequence.
Note that if defined, $B_{x,i}$ contains exactly $i$ vertices, $B_{x,i+1}$ is a superset of $B_{x,i}$, and $B_{x,i}$ maximizes $\leftb(\tspan(X))$ among all subsets $X \subseteq V_x$ with $\norm{X}=i$.

We say that $D$ is a \emph{$k$-block} defense if it counters every attack $B_{x,i}$ for every possible right endpoint $x$ and every $i \le k$.
\begin{lemma}\label{lem:block}
A proper $k$-block defense~$D$ counters every $k$-attack in $G$.
\end{lemma}
\begin{proof}
    Assume to the contrary that a proper $k$-block defense $D$ does not counter some $k$-attack.
    By \cref{obs:hall} we have a set $A$ with $\setcount_D(N[A]) < \norm{A} = m \le k$.
    Among such sets, we select $A$ with the minimum size $m$.

    Let $w \in A$ be the vertex that maximizes $\rightb(I_w)$ among the vertices in $A$.
    Let $x = \rightb(I_w)$, $c=\norm{V_x}$, and $v_1,v_2,\ldots,v_c$ be $V_x$ arranged in a sequence sorted descending by the left endpoint of the representing interval.
    All elements of $A$ are represented to the left of $x$, so we have $m \le c$.
    Recall that $B_{x,i} = \set{v_1,v_2,\ldots,v_i}$ for every $1 \le i \le c$.
    By the assumption of the lemma, we know that $A \neq B_{x,m}$.
    Since $D$ counters $B_{x,m}$, we have $\setcount_D(N[A]) < \setcount_D(N[B_{x,m}])$ and there is $d \in D$ with $d \in N[B_{x,m}]$ and $d \notin N[A]$. 
    Every interval representing a vertex in $A$ is either completely to the left or completely to the right of $I_d$.
    Let
    $$A_1 = \{v \in A: \rightb(I_{v}) < \leftb(I_d)\} \quad \text{and} \quad A_2 = \{v \in A: \rightb(I_d) < \leftb(I_{v})\}\text{.}$$
    Since $x \in I_{w}$ and $x \notin I_d$, we have $w \in A_2$.
    Since $\leftb(\tspan(A)) < \leftb(\tspan(B_{x,m}))$, we have $A_1 \neq \emptyset$.
    We have partitioned $A$ into two nonempty subsets $A_1$ and $A_2$.

    By the minimality of $A$ we have $\setcount_D(N[A_1]) \ge \norm{A_1}$ and $\setcount_D(N[A_2]) \ge \norm{A_2}$.
    We get that there is at least one $d' \in D \cap N[A_1] \cap N[A_2]$, as otherwise we would have $\setcount_D(N[A]) = \setcount_D(N[A_1]) + \setcount_D(N[A_2]) \ge \norm{A_1} + \norm{A_2} = \norm{A}$.
    Interval $I_{d'}$ intersects both $\tspan(A_1)$ and $\tspan(A_2)$.
    This allows us to write the following inequalities:
    $$
    \leftb(I_{d'}) < \rightb(\tspan(A_1)) < \leftb(I_d) < \rightb(I_d) < \leftb(\tspan(A_2)) < \rightb(I_{d'})\text{,}
    $$
    and conclude that $I_d$ is a proper subset of $I_{d'}$ contradicting with $D$ being a proper defense.
\end{proof}

We are ready for the proof of the main result of this section.
\thminterval*
\begin{proof}
The proposed algorithm, see \cref{alg:greedy} for the pseudocode, works as follows.
It assumes that $G$ has $n$ vertices and is given in the interval representation by the intervals $I_1,I_2,\ldots,I_n$.
It sorts the intervals ascending by their right endpoints.
Then, for every $1 \le i \le n$, let $x$ be the right endpoint of the $i$-th interval in the list.
The algorithm calculates the number $count$ of defenders missing to counter every $m$-block attack $B_{x,m}$ for $m \le k$.
Then, it selects the interval $d$ that maximizes the right endpoint among all intervals that include $x$.
The algorithm adds $count$ copies of the interval $d$ to $D$.

\begin{algorithm}
    \caption{Greedy Multiset Defensive Domination in Interval Graphs}
    \label{alg:greedy}
    \begin{algorithmic}[1] %
        \Procedure{GreedyDefense}{$\brac{I_1,I_2,\ldots,I_n},k$} %
            \State $D \gets \emptyset$
            \State $I_1,I_2,\ldots,I_n \gets$ intervals $I_1,I_2,\ldots,I_n$ sorted ascending by their right endpoints
            \For{$i=1,2,\ldots,n$}
                \State $x = \rightb(I_i)$
                \For{$m=1,2,\ldots,\min(i,k)$}
                    \State $count \gets \max(0,m-\setcount_D(N[B_{x,m}]))$
                    \State $d \gets$ interval with $\leftb(d) < x$ and maximum $\rightb(d)$
                    \State Add $count$ copies of $d$ to $D$
                \EndFor
            \EndFor
            \State \textbf{return} $D$
        \EndProcedure
    \end{algorithmic}
\end{algorithm}

Clearly, \cref{alg:greedy} runs in polynomial time.
The algorithm only adds inclusion maximal intervals to $D$, so the constructed multiset defense $D$ is proper. 
The algorithm explicitly adds the required number of defenders to counter every possible attack $B_{x,m}$, so $D$ is also $k$-block.
Thus, by \cref{lem:block} we get that $D$ counters every $k$-attack in $G$.

For the proof that the defense $D$ is of minimum size, let $\ell=\norm{D}$, $J_1,J_2,\ldots,J_{\ell}$ be the multiset of intervals in $D$ sorted ascending by their left endpoints.
Assuming to the contrary, let $D'$ be another proper defense that counters every $k$-attack with $\ell' =\norm{D'} < \norm{D}$ and $K_1,K_2,\ldots,K_{\ell'}$ be the multiset of intervals in $D'$ sorted ascending by their left endpoints.
Among such defenses, select one that maximizes the number $p$ such that $J_i = K_i$ for all $1\le i \le p$.
We have $p < \ell' < \ell$.
Let $d = J_{p+1}$ and $d' = K_{p+1}$.
We do not have $d' \subsetneq d$, as otherwise we could exchange $d'$ for $d$ in $D'$ and get another defense that counters every $k$-attack of the same size, but with larger $p$.
We also do not have $d \subsetneq d'$, as the algorithm only adds inclusion maximal intervals to $D$.
Let $x,m$ be such that the algorithm added $d$ to $D$ when considering attack $B_{x,m}$.

If $\leftb(d') > \leftb(d)$ then $\rightb(d') > \rightb(d)$ and as the algorithm did not add $d'$ to $D$ we have $\leftb(d') > x$.
This means that $D'$ has exactly $p$ intervals $K_1,\ldots,K_p = J_1,\ldots,J_p$ with the left endpoint less than or equal to $x$.
The algorithm calculated that these intervals do not counter $B_{x,m}$ and no other interval in $D'$ intersects $\tspan(B_{x,m})$.
We conclude that $D'$ does not counter $B_{x,m}$, a contradiction.

If $\leftb(d') < \leftb(d)$ then $\rightb(d') < \rightb(d)$ and we can exchange $d'$ for $d$ in $D'$ and get another defense $D''$ that counters every $k$-attack with $\norm{D''} = \norm{D'}$ but with bigger $p$.
Indeed, we can show that the resulting defense $D''$ satisfies $\setcount_{D''}(N[B_{y,q}]) \ge q$ for every possible right endpoint $y$ and every $q \le k$.
For $y < x$, this follows from the fact that $D''$ agrees with $D$ on the first $p$ intervals that are enough to counter these attacks.
For $y \ge x$, we have $d' \in N[B_{y,q}] \implies d \in N[B_{y,q}]$ and $\setcount_{D''}(N[B_{y,q}]) \ge \setcount_{D'}(N[B_{y,q}]) \ge q$.
We get a contradiction with the choice of $D'$.
\end{proof}

We claim that \cref{alg:greedy} can be implemented to run in $\Oh{nk}$ time using standard techniques, and we omit the details of this implementation.

\section{Summary}

In this work, we have shown that both \Pdefdom and \Pmultidom are \SPTWO-complete.
For the multiset variant of the problem, in the class of interval graphs, we have indicated a polynomial-time algorithm.
This algorithm does not work in the original setting where at most one defender can be located at a single vertex.
Furthermore, we do not even know if \Pcheckdom admits a polynomial-time algorithm in the class of interval graphs.
We would like to see the complexity status of \Pdefdom resolved in the class of interval graphs.

We also believe that potential applications in the facility location problem should justify the investigation of defensive domination problems in the class of planar graphs.
Note that the reductions presented for \SPTWO-completeness of \Pdefdom and \Pmultidom or the \W{1}-hardness of \Pcocheckdom construct graphs with large cliques and cannot be applied to show the hardness in the class of planar graphs.
There is a natural dynamic programming \FPT-algorithm that checks \Pcocheckdom when parametrized by the tree-width of the input graph.
When looking for dangerous $k$-attacks against a fixed defense, it is enough to consider attacks $A$ such that $N[A]$ is connected.
This means that in a planar graph, we can fix some outerplanar decomposition of the graph and only consider $k$-attacks that span at most $2k-1$ adjacent layers of the outerplanar decomposition.
As the tree-width of such subgraphs is bounded, we obtain a simple \FPT-algorithm that checks \Pcocheckdom when parametrized by $k$ in the planar graphs.
This shows that in the parametrized sense, defensive domination problems might be easier in planar graphs than they are in general graphs.
This motivates the following questions.
Does \Pcheckdom admit a polynomial-time algorithm in the class of planar graphs?
What is the complexity of \Pdefdom and \Pmultidom in the class of planar graphs?
Both problems are \NP-hard, but we do not know if they are \SPTWO-complete.

\bibliography{defensive_domination}

\appendix

\clearpage
\section{Clique Node Deletion}\label{sec:clique}

The main result of this paper is proved by a reduction from \Pcliquedel to \Pdefdom.
Clique Node Deletion Problem was first introduced by Rutenburg~\cite{Rutenburg86} in a more general setting called Generalized Node Deletion.
Rutenburg gives an idea for a proof~\cite[Theorem~6]{Rutenburg86} that \Pcliquedel is \SPTWO-complete.
As \Pcliquedel is an important intermediate problem for our result, we present a streamlined proof based on Rutenburg's idea.

The proof is based on a reduction from the following problem, which is a variation on the quantified boolean formula satisfaction problem.
It is a natural \SPTWO-complete problem with an easy reduction from \PSPTWOSAT~\cite{Stockmeyer76}.
\newprobspec{\PSPTWO}{$3$-CNF formula $\varphi(x_1,\ldots,x_a,y_1,\ldots,y_b)$ with variables in two disjoint sets $\set{x_1,\ldots,x_a}$ and $\set{y_1,\ldots,y_b}$}{$\mathsf{Yes}$ if and only if the following Boolean formula is true.
    $$
        \exists_{x_1,x_2,\ldots,x_a}:\quad \neg \exists_{y_1,y_2,\ldots,y_b}:\quad \varphi(x_1,\ldots,x_a,y_1,\ldots,y_b)
    $$}

\input ./figures/rutenburg2.tex

We are now ready to present the proof of the following theorem.
\thmrutenburg*
\begin{proof}
    We reduce from \PSPTWO to \Pcliquedel.
    Assume that we are given two disjoint sets of variables $X=\set{x_1,\ldots,x_a}$, $Y=\set{y_1,\ldots,y_b}$, and a set of clauses $\mathcal{C}=\set{C_1,\ldots,C_c}$ with each clause having exactly three occurrences of three distinct variables in $X \cup Y$.
    For technical reasons, we assume that $c>6$, as otherwise there are at most $18$ variables and we can simply check all the possibilities.
    We set $s=ac+3c$, $t=b+c$, and construct a graph $G$ such that $G$ admits a subset~$X$ of vertices with $\norm{X}\le s$ such that $G\setminus X$ does not contain $K_t$ as a subgraph if and only if the following formula
    $$
        \exists_{x_1,x_2,\ldots,x_a}:\quad \neg \exists_{y_1,y_2,\ldots,y_b}:\quad C_1 \wedge C_2 \wedge \ldots\wedge C_c
    $$
    is true.
    The graph~$G$ is constructed in the following steps.
    Consult Figure~\ref{fig:rutenburg} for an example.
    \begin{enumerate}
        \item For each variable $x_i$ ($1 \leq i\leq a)$ we introduce $G_{x_i}$: a copy of $K_{c,c}$ with one bipartition class called \emph{positive} and the other class called \emph{negative}.
        We number positive vertices from $1$ to $c$ and number negative vertices from $1$ to $c$.
        \item For each variable $y_j$ ($1 \leq j \leq b$) we introduce $G_{y_j}$: two independent vertices, one called \emph{positive} and the other called \emph{negative}.
        \item For each clause $C_k$ ($1 \leq k \leq c$) we introduce $G_{C_k}$: a copy of $K_{3,3}$ with one bipartition class called \emph{good} and the other class called \emph{bad}.
        We select one of the bad vertices to be \emph{ugly}.
        In~$C_k$ there are $3$ occurrences of variables.
        For each such occurrence, we assign a distinct good vertex in $G_{C_k}$.
        We call the good vertices assigned to variables in $X$ ($Y$) to be \emph{$X$-good} (\emph{$Y$-good}).
        \item \label{item:choice} For each edge $(u,v)$ in every $G_{x_i}$ and every $G_{C_k}$ we add a copy of $K_t$ on $u$, $v$ and additional $t-2$ new vertices.
        \item \label{item:good} For each $G_{C_k}$ with $g$ $X$-good vertices we compose set $Z$ of all $X$-good vertices in $G_{C_k}$ and every positive (negative) vertex numbered $k$ in $G_{x_i}$ with $x_i$ having a positive (negative) occurrence in $C_k$.
        Set $Z$ has exactly $2g$ elements.
        We add a copy of $K_{t-1+g}$ on vertices in $Z$ and additional $t-1-g$ new vertices.
        We call this added clique $Q_{C_k}$.
        \item \label{item:goody} For each positive (negative) occurrence of $y_j$ in $C_k$, let $u$ be the $Y$-good vertex in $G_{C_k}$ assigned to this occurrence.
        We add an edge between $u$ and the positive (negative) vertex in $G_{y_j}$.
        We also add edges between $u$ and both vertices in every other $G_{y_{j'}}$ for $j' \neq j$.
        \item \label{item:uglyy} For each clause $C_k$, let $u$ be the ugly vertex in $G_{C_k}$.
        We add edges between $u$ and both vertices in every $G_{y_{j}}$.
        \item \label{item:yy} For every $1\leq j < j' \leq b$, we add all four edges between any of the two vertices in $G_{y_{j}}$ and any of the two vertices in $G_{y_{j'}}$.
        \item \label{item:cross} Finally, for every $1\leq k < k' \leq c$, every bad or $Y$-good vertex $u$ in $G_{C_k}$, and every bad or $Y$-good vertex $u'$ in $G_{C_{k'}}$ we add an edge between $u$ and $u'$.
        
    \end{enumerate}
    Observe that $G$ contains many copies of $K_t$. 
    We distinguish four types of such cliques:
    \begin{enumerate}[label=(\Alph*)]
        \item \label{item:clique_x} A clique that contains a positive and a negative vertex in a single $G_{x_i}$.
        The common neighborhood of such a pair of vertices is the set of additional vertices added in step~\ref{item:choice} of the construction.
        \item \label{item:clique_c} A clique that contains a good and a bad vertex in a single $G_{C_k}$.
        The common neighborhood of such a pair of vertices is the set of additional vertices added in step~\ref{item:choice} of the construction.
        \item \label{item:clique_g} A clique that contains an $X$-good vertex in some $G_{C_k}$ and is not of type~\ref{item:clique_c}.
        As this clique does not contain any bad vertex in $G_{C_k}$, it must be a subclique of $Q_{C_k}$ constructed in step~\ref{item:good}.
        \item \label{item:clique_u} Other. Observe that every other clique of size $t$ can contain only $Y$-good or ugly vertices in $G_{C_j}$ or positive or negative vertices in $G_{y_i}$.
    \end{enumerate}
    We independently prove the implications in both directions.
    
    ($\Leftarrow$)
    Let $\nu_1,\nu_2,\ldots,\nu_a$ be a valuation of $x_1,x_2,\ldots,x_a$ such that formula $$\exists_{y_1,y_2,\ldots,y_b}:\quad C_1 \wedge C_2 \wedge \ldots \wedge C_c$$ is false.
    We construct a set $X$ of vertices in $V(G)$ in the following way.
    From each $G_{x_i}$ we select all the $c$ negative vertices if $\nu_i$ is true and all $c$ positive vertices if $\nu_i$ is false.
    For each $C_k$ if $C_k$ is satisfied by any of the variables $x_i=\nu_i$ we select all $3$ good vertices.
    Otherwise, we select all $3$ bad vertices.
    The constructed set $X$ has $ac+3c$ vertices.
    We claim that $G\setminus X$ does not contain $K_t$ as a subgraph.
    Each clique of type~\ref{item:clique_x} or~\ref{item:clique_c} is removed by the selection of vertices.
    Each clique $Q_{C_k}$ of type~\ref{item:clique_g} corresponding to a clause $C_k$ with $g$ occurrences of variables in $X$ is removed, as we either remove $g$ vertices corresponding to the false occurrences of variables in variable gadgets, or all good vertices in clause gadgets $G_{C_k}$.
    
    Now, assuming to the contrary, let $Q$ be some clique of type~\ref{item:clique_u} that remains in $G\setminus X$.
    It is clear that $Q$ can have at most one vertex in each $G_{y_j}$ and at most one vertex in each $G_{C_k}$.
    As $t=b+c$, $Q$ has exactly one vertex in each $G_{y_j}$ that corresponds to a valuation $\mu_1,\mu_2,\ldots,\mu_b$ of variables $y_1,y_2,\ldots,y_b$.
    We claim that the combined valuation ($x_i=\nu_i$, $y_j=\mu_j$) satisfies all clauses $C_1,C_2,\ldots,C_c$, which gives a contradiction.
    As clique $Q$ also has exactly one vertex in each $G_{C_k}$, let $u$ be that vertex and observe that $y$ is ugly or $Y$-good in $G_{C_k}$.
    If $u$ is the ugly vertex in $G_{C_k}$, it means that the bad vertices are not removed, and $C_k$ is satisfied by one of the variables $x_1,x_2,\ldots,x_a$.
    If $u$ is a $Y$-good vertex, then $u$ corresponds to an occurrence of some variable $y_j$.
    The construction of $G$ guarantees that there is only one edge between $u$ and a vertex in $G_{y_j}$ that corresponds to a valuation of $y_j$ that satisfies clause $C_k$.
    Thus, every clause is satisfied by some variable.
    A contradiction.

    ($\Rightarrow$)
    Let $X$ be a subset of vertices of $G$ witch $\norm{X} \leq ac+3c$ and $G\setminus X$ does not contain $K_t$ as a subgraph.
    We construct a valuation $\nu_1,\nu_2,\ldots,\nu_a$ of variables $x_1,x_2,\ldots,x_a$ such that formula $$\exists_{y_1,y_2,\ldots,y_b}:\quad C_1 \wedge C_2 \wedge \ldots \wedge C_c$$ is false.
    To remove all cliques of type~\ref{item:clique_x}, $X$ must include at least $a$ positive or $a$ negative vertices in each $G_{x_i}$.
    To remove all cliques of type~\ref{item:clique_c}, $X$ must include at least $3$ good or $3$ bad vertices in each $G_{C_k}$.
    As $\norm{X} \leq ac+3c$, we get that $X$ includes exactly $a$ positive or exactly $a$ negative vertices in each $G_{x_i}$.
    We set $\nu_i$ to be true if and only if $X$ includes negative vertices in $G_{x_i}$.
    Observe that if the constructed valuation of $x_1,x_2,\ldots,x_a$ satisfies clause $C_k$, then $X$ has to include $3$ good vertices in $G_{C_k}$, and the ugly vertex in $G_{C_k}$ remains in $G\setminus X$.
    Otherwise, there would remain a clique of type~\ref{item:clique_g} in $G\setminus X$.

    Now, assuming to the contrary, let $\mu_1,\mu_2,\ldots,\mu_b$ be a valuation of variables $y_1,y_2,\ldots,y_b$ that satisfies all the clauses.
    We construct a clique $Q$ in $G\setminus X$ of size $t=b+c$ the following way.
    For each $1\leq j \leq b$, we select a positive (negative) vertex from $G_{y_j}$ when $\mu_j$ is true (false).
    For each $1 \leq k \leq c$, we select the ugly vertex from $G_{C_k}$ if it is not removed.
    If the ugly vertex is removed, then we know that $C_k$ is satisfied by one of the variables $y_1,y_2,\ldots,y_b$, and we select the $Y$-good vertex in $G_{C_k}$ that was assigned to the satisfied occurrence.

    We claim that the resulting set of vertices induces a clique with $b+c$ vertices.
    Indeed, vertices from different subgraphs $G_{y_j}$, $G_{y_{j'}}$ are always connected by an edge (step~\ref{item:yy}).
    All ugly vertices are connected to each other (step~\ref{item:cross}) and both vertices in every subgraph $G_{y_j}$ (step~\ref{item:uglyy}).
    Let $u$ be a $Y$-good vertex selected from some $G_{C_k}$ to $Q$.
    Vertex $u$ corresponds to the occurrence of variable $y_j$ in $C_k$.
    Vertex $u$ is connected to ugly and $Y$-good vertices in other $G_{C_{k'}}$ for $k\neq k'$ (step~\ref{item:cross}) and both vertices in every subgraph $G_{y_{j'}}$ for $j\neq j'$ (step~\ref{item:goody}).
    As $u$ corresponds to a satisfied occurrence of $y_j$ in $C_k$, we get that $u$ is also connected to the vertex selected from $G_{y_j}$.
    Thus, we have found a clique $K_t$ in $G\setminus X$.
    A contradiction.
\end{proof}
\end{document}